\documentclass{llncs}
\usepackage[english]{babel}
\usepackage{mathtools,amssymb}
\usepackage {xspace}
\usepackage {url}
\usepackage {plaatjes}
\usepackage {times}
\usepackage {xcolor}
\usepackage {microtype}
\usepackage[vlined,ruled]{algorithm2e}
\usepackage{graphicx,amssymb,amsmath,verbatim,wrapfig}
\usepackage{a4wide}
\usepackage{xcolor}
\usepackage{plaatjes}

\newcommand{\eps}{\varepsilon}

\newcommand{\MST}{\mathit MST}
\newcommand{\SMT}{\mathit SMT}
\newcommand{\pp}{\ensuremath{k}}
\newcommand{\etal}{\emph{et~al.}}

\graphicspath{{./figures/}}

\definecolor{purple}{hsb}{.77,1,1}
\newcommand{\niceremark}[3]{\textcolor{red}{\textsc{#1 }}\textcolor{blue}{\textsc{#2: }}\textcolor{purple}{\textsf{#3}}}

\newcommand{\rodrigo}[2][says]{\niceremark{Rodrigo}{#1}{#2}}
\newcommand{\marc}[2][says]{\niceremark{Marc}{#1}{#2}}

\let\doendproof\endproof
\renewcommand\endproof{~\hfill$\square$\doendproof}

\renewcommand{\paragraph}[1]{\medskip\noindent\textbf{#1.}}

\title{Colored Spanning Graphs for Set Visualization\thanks{
A preliminary version of this paper appeared in the proceedings of the {\em 21st International Symposium on Graph Drawing}~\cite{hkvlsss-csgsv-13}.}}

\author
{
  Ferran Hurtado\inst{1}
  \and Matias Korman\inst{2}
  \and Marc van Kreveld\inst{3}
  \and Maarten L\"offler\inst{3}
  \and Vera Sacrist\'an\inst{1}
  \and \\ Akiyoshi Shioura\inst{4}
  \and Rodrigo I. Silveira\inst{1}
  \and Bettina Speckmann\inst{5}
  \and Takeshi Tokuyama\inst{2}
  }

\institute{Dept. de Matem\`atiques, Universitat Polit\`ecnica de Catalunya, Spain \email{\{vera.sacristan|rodrigo.silveira\}@upc.edu}
\and Tohoku University, Japan \email{\{mati|tokuyama\}@dais.is.tohoku.ac.jp}
\and Dept. of Computing and Information Sciences, Utrecht University, the Netherlands \email{\{m.j.vankreveld|m.loffler\}@uu.nl}
\and Dept. of Social Engineering, Tokyo Institute of Technology, Japan \email{shioura.a.aa@m.titech.ac.jp}
\and Dept. of Mathematics and Computer Science, Technical University Eindhoven, the Netherlands \email{b.speckmann@tue.nl}
}

\begin{document}

\maketitle

\begin{abstract}
We study an algorithmic problem that is motivated by ink minimization for sparse set visualizations. Our input is a set of points in the plane which are either blue, red, or purple. Blue points belong exclusively to the blue set, red points belong exclusively to the red set, and purple points belong to both sets.
A \emph{red-blue-purple spanning graph} (RBP spanning graph) is a set of edges connecting the points such that
the subgraph induced by the red and purple points is connected,
and the subgraph induced by the blue and purple points is connected.

We study the geometric properties of minimum RBP spanning graphs and the algorithmic problems associated with computing them. Specifically, we show that the general problem can be solved in polynomial time using matroid techniques. In addition, we discuss more efficient algorithms for the case in which points are located on a line or a circle, and also describe a fast $(\frac 12\rho+1)$-approximation algorithm, where $\rho$ is the Steiner ratio.\end{abstract}

\setcounter{footnote}{0}
\subsection*{In memoriam: Ferran Hurtado (1951--2014)}
This paper was initiated during a research visit hosted by Ferran and his group in Barcelona. Following tradition, hours of research were complemented by relaxing evenings with great food and wine. The authors would like to thank Ferran for being a supportive mentor, an inspiring colleague, and a great friend.

\section{Introduction}

Visualizing sets and their elements is a recurring theme in information visualization (see the recent state-of-the-art report by Alsallakh \etal~\cite{Alsallakh2014}). Sets arise in many application areas, as varied as social network analysis (grouping individuals into communities), linguistics (related words), or geography (related places). Among the oldest representations for sets are Venn diagrams~\cite{Edwards2004Cogwheels} and their generalization, Euler diagrams. These representations are natural and effective for a small number of elements and sets. However, for larger numbers of sets and more complicated intersection patterns more intricate solutions are necessary. The last years have seen a steady stream of developments in this direction, both for the situation where the location of set elements can be freely chosen and for the important special case that elements have to be drawn at particular fixed positions (for example, restaurant locations on a city map).

Our paper is motivated by some recent approaches that use very sparse enclosing shapes when depicting sets. LineSets~\cite{alper11} are the most minimal of all, reducing the geometry to a single continuous line per set which connects all elements. Both Kelp Diagrams~\cite{kelp} and its successor KelpFusion~\cite{kelpfusion} are based on sparse spanning graphs, essentially variations of minimal spanning trees for different distance measures. These methods attempt to reduce visual clutter by reducing the amount of ``ink'' (see Tufte's rule~\cite{tufte}) necessary to connect all elements of a set. However, although the results are visually pleasing, neither method does use the optimal amount of ink. In this paper we explore the algorithmic questions that arise when computing spanning graphs for set visualization which are optimal with respect to ink usage.

\smallskip\noindent
{\bf Problem statement.} Our input is a set of $n$ points in the plane. Each point is a member of one or more sets. We mostly consider the case where there are exactly two sets, namely a red and a blue set. A point is red if it is part only of the red set and it is blue if it is part only of the blue set. A point that is part of both the red and the blue set is purple.

\begin{wrapfigure}[10]{r}{.49\textwidth}
\centering
\vspace{-1.5\baselineskip}
\includegraphics{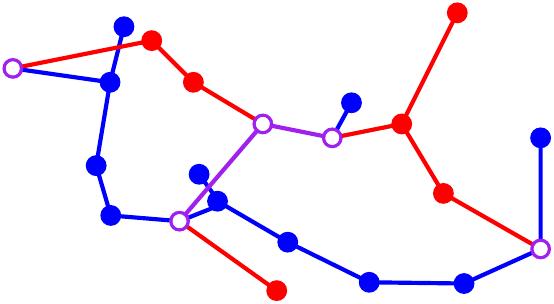}
\vspace{-.25\baselineskip}
\caption{A minimum RBP spanning graph.}
\label{fig:firstexample}
\end{wrapfigure}
A \emph{red-blue-purple spanning graph} (RBP spanning graph) for a set of points that are red, blue, and purple is a set of edges connecting the points such that
the subgraph induced by the red and purple points is connected,
and the subgraph induced by the blue and purple points is connected (see Figure~\ref{fig:firstexample}).
A \emph{minimum RBP spanning graph} for a set of points
that are red, blue, and purple is a red-blue-purple spanning graph
that has minimum weight (total edge length).
In this paper we consider the algorithmic problems associated with computing minimum RBP spanning graphs.

\smallskip\noindent
{\bf Results and organization.} We first review related work. In Section~\ref{sec:properties} we describe and prove various useful properties of (minimum) RBP spanning graphs. Then, in Section~\ref{sec:linecircle}, we consider the two special cases where the points are located on a line or on a circle. This setting is meaningful if the elements of the sets are not associated with a specific location (for example, social networks or software systems). Here visualizations frequently choose to arrange elements in simple configurations such as lines or circles.
We give an $O(n)$ time algorithm for points on a line, assuming that the input is already sorted. For points on a circle we exploit a structural result which allows us to use dynamic programming in $O(k^3 + n)$ time, where $k$ is the number of purple points. In Section~\ref{sec:approx} we turn to approximations. We describe an $O(n \log n)$ time algorithm that computes a $(\frac 12\rho+1)$-approximation of the minimum RBP spanning graph, where $\rho$ is the Steiner ratio. In Section~\ref{sec:matroid} we give a general algorithm for computing a minimum RBP spanning graph. The algorithm runs in $O(n^6)$ time and is based on matroid optimization techniques.\footnote{We note that in the conference version of this paper~\cite{hkvlsss-csgsv-13}, we claimed that the problem was NP-hard. However, our NP-hardness reduction turned out to be incorrect.} Finally, in Section~\ref{sec:extensions} we discuss various extensions for situations with more than two sets.

\smallskip\noindent
{\bf Related work.} In recent years a number of papers explored the problem of automatically drawing Euler diagrams, for example, Simonetto and
Auber~\cite{Simonetto2008UndrawableEulerDiagrams}, Stapleton
\etal~\cite{Stapleton2011InductivelyGeneratingEulerDiagrams}, and Henry Riche and Dwyer~\cite{Riche2010UntanglingEulerDiagrams}. These methods assume that the locations of the set elements can be freely chosen. An important variation is the case that elements have to be drawn at fixed positions. Collins \etal~\cite{bubblesets} presented \emph{Bubble Sets}, a method based on isocontours.  A similar approach was suggested by Byelas and Telea~\cite{Byelas2008AreasOfInterest}.
\emph{LineSets}, by Alper \etal~\cite{alper11}, attempt to improve the overall readability by the minimalist approach of drawing a single line per set. Dinkla \etal~\cite{kelp} introduced \emph{Kelp Diagrams}, which use a sparse spanning graph, essentially a minimum spanning tree with some additional edges. Kelp Diagrams were extended by Meulemans \etal~\cite{kelpfusion} to a hybrid technique named \emph{KelpFusion} which uses a mix of hulls and lines to generate fitted boundaries.

Sets defined over points in the plane can be interpreted as an
embedding of a \emph{hypergraph} where the points are vertices and each set is
a hyperedge connecting an arbitrary number of vertices. Drawings of hypergraphs have been
discussed in several papers, for example, 
by Kaufmann \etal~\cite{kaufmann09}, Buchin \etal~\cite{Buchin2011PlanarSupports}, and
 Brandes \etal~\cite{Brandes2010PathBasedSupportsJ}.

Also in the area of discrete and computational geometry
many problems on colored point sets have been studied. Possibly the most famous example is the \emph{Ham-Sandwich Theorem}: given a set of $2n$ red points and $2m$ blue
points so that no three points are aligned, there is always a line $\ell$ such that each open halfplane bounded by $\ell$ contains exactly $n$ red points and $m$ blue points.
There have been many variations on this theorem
and also many other results on finding configurations or geometric graphs with  constraints depending on colors. We refer the interested reader to the survey by Kaneko and Kano~\cite{KK03}, and to Chapter 8 in the collection of research problems \cite{BMP05}.

From an algorithmic point of view, many problems have been considered, here we mention only a few of them. The \emph{bichromatic closest pair} (e.g., see Preparata and Shamos~\cite{Pre85}~Section 5.7, Agarwal \etal~\cite{AES91}, and Graf and Hinrichs \cite{GH93}), the \emph{chromatic nearest
neighbor search} (see Mount et al. \cite{MNSW00}), the problems on finding \emph{smallest color-spanning objects} (see Abellanas \etal~\cite{AHIKLMPS01}),
the \emph{colored range searching} problems (see Agarwal \etal~\cite{AGM02}), and the \emph{group Steiner tree} problem where, for
a graph with colored vertices, the objective is to find a minimum weight subtree
that covers all colors (see Mitchell \cite{Mitchell98}, Section 7.1).
Borgelt \etal~\cite{BKLLMSV09} discuss computing planar red-blue minimum spanning trees where edges may connect red and blue points only, and the tree must be planar.
Finally, Tokunaga~\cite{Tokunaga1996} considers a set of red or blue points in the plane and computes two geometric spanning trees of the blue and the red points such that they intersect in as few points as possible.

\section{Properties of RBP spanning graphs}\label{sec:properties}

For any set $S$ of points, let $MST(S)$ denote a Euclidean minimum spanning tree of $S$. Given a graph $G$ on vertex set $S$, and a set $S'\subseteq S$, we denote by $G[S']$ the subgraph of $G$ induced by $S'$. Given two graphs $G_1=(S_1,E_1)$, and $G_2=(S_2,E_2)$ with possibly different vertex sets, we define their intersection as $G_1\cap G_2=(S_1\cap S_2, E_1\cap E_2)$. Given three sets $R$, $B$, and $P$ of respectively red, blue, and purple points, let $G^*(R,B,P)$ denote a minimum RBP spanning graph. When the tuple $(R,B,P)$ is clear from the context, we simply denote it by $G^*$.

We say that an edge of $G^*$ is \emph{red} if it connects two red
points or a red and a purple point. We call an edge \emph{blue} if it connects
two blue points or a blue and a purple point. We call an edge \emph{purple} if it
connects two purple points. A minimum RBP spanning graph $G^*$ does not contain edges between a red and a blue point.
The subgraph induced by the red and purple points in a minimum
RBP spanning graph is a spanning tree (since otherwise we could remove an edge from $G^*$ and reduce its weight). Naturally, the analogous statement
holds for the blue and purple points. Figure~\ref{fig:spanningexamples} (a) illustrates the trade-off between red, blue, and purple edges in a minimum RBP spanning graph.

When there are no ties (e.g., if no two edges of the same color have the same length), every red edge in a minimum RBP spanning graph also occurs in a minimum spanning tree of only the red and purple points (since otherwise we could replace it by another red edge of smaller weight). The corresponding statement is true also for the blue edges, but not for the purple edges. That is, there can be purple edges in a minimum RBP spanning graph that do not occur in a minimum spanning tree of only the purple points. In fact, in the upcoming Proposition~\ref{prop:MartiniGlass} we show something stronger: we give a problem instance whose RBP spanning graph has a purple edge that crosses several other purple edges.

Just as with standard minimum spanning trees, the degree of the points in a minimum RBP spanning graph is bounded.

\begin{lemma}
The maximum degree of a point in a minimum RBP spanning graph is at most $18$.
\end{lemma}
\begin{proof}
We now proceed to prove the statement by contradiction. Assume that there exists a problem instance with a solution $G^*$ in which a point $p$ is adjacent to more than $18$ points. By the pigeonhole principle, we conclude that $p$ is adjacent to at least $7$ points of the same color. Let $N$ be the set of such vertices, and observe that the edges in $G^*[N \cup \{p\}]$ form a star. Thus, if we replace these edges by those in $MST(N \cup \{p\})$ we obtain a different graph that spans exactly the same vertices and has smaller weight.
\end{proof}

We also note that the above bound is tight: let $p$ be a purple point located in the center of a regular hexagon with radius $3$ and two more regular hexagons with radius $1$, one slightly rotated. Place a purple point on each corner of the large hexagon, place red points on the corners of one of the smaller hexagons, and blue points on the corners of the other one.
Then the star graph with $p$ at the center (with degree 18) is a minimum RBP spanning graph (see Figure~\ref{fig:spanningexamples}~(c)).

It is worth mentioning that the degree-$18$ vertex in the previous example can be avoided, since the same point set has other RBP spanning graphs where the maximum degree is $15$.
In fact, a degree higher than $15$ is never necessary: even if we allow degeneracies, there is a way to break ties so that no vertex of a minimum spanning tree has degree larger than $5$~\cite{Wu20061}. Thus, if a vertex is adjacent to more than $5$ vertices of the same color, we can remove one of the edges and replace it by one of equal or smaller length, preserving connectivity. Similarly, we can show that a red or blue point can have degree $6$ in a minimum RBP spanning graph, but for any point set there exists a minimum RBP spanning graph in which no red or blue point has degree larger than $5$.
At the same time, a vertex of degree $15$ may be sometimes necessary:
if we replace the hexagons in the previous construction by pentagons, we obtain a similar construction whose \emph{unique} minimum RBP spanning graph has a point of degree $15$.

It is easy to see that a minimum RBP spanning graph is not necessarily planar.
Red and blue edges are mostly independent and they can easily cross. Moreover, a red and a purple edge can cross, a blue and a purple edge can cross, and even two purple edges can cross in a minimum RBP spanning graph, as shown in Figure~\ref {fig:spanningexamples} (b). In fact, a single purple edge can cross arbitrarily many purple edges. 

\begin{figure}[t]
\centering
\includegraphics{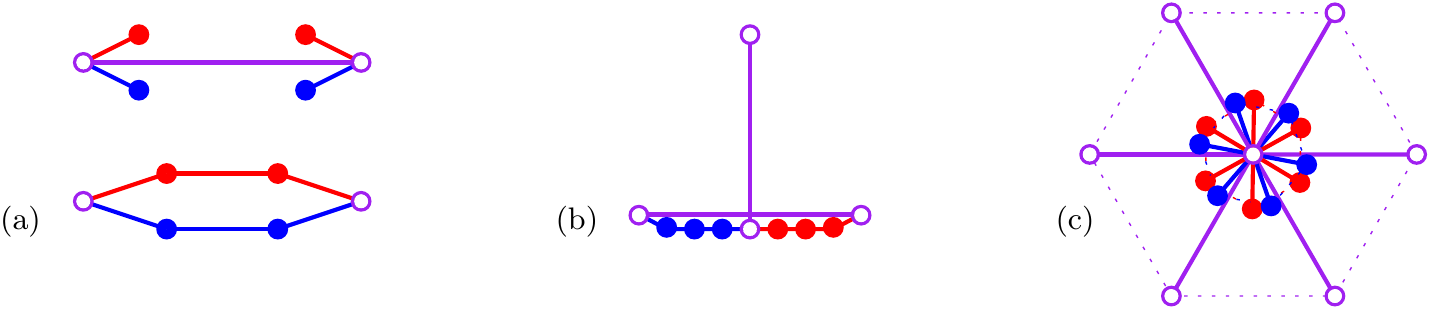}
\caption{(a) Two examples of minimum RBP spanning graphs of similar
configurations of points.
(b) A minimum RBP spanning graph with a purple edge crossing. (c) A minimum RBP spanning graph with a vertex of degree 18.}
\label{fig:spanningexamples}
\end{figure}

\begin{proposition}
\label{prop:MartiniGlass}
A purple edge in an optimal RBP spanning graph can cross $\Theta(n)$ other purple edges.
\end{proposition}

The proof of this statement is constructive: we create a problem instance whose (unique) RBP spanning graph has a purple edge crossed by all other purple edges. Due to its length, the exact details of the construction and the additional observations necessary to prove this statement are relegated to the appendix. A sketch of the construction is shown in Figure~\ref{fig:MartiniGlass-Preview}.

\begin{figure}[tb]
\centering
\includegraphics{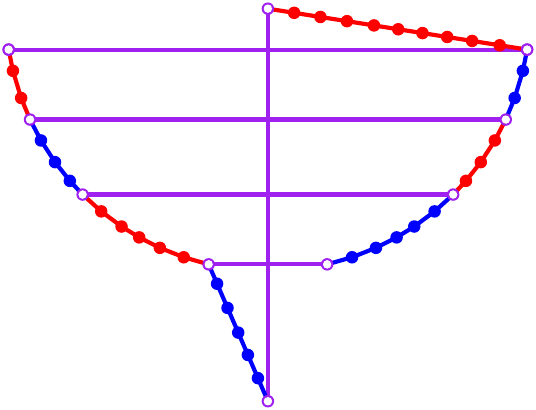}
\caption{Schematic view of a set of points whose minimum RBP spanning graph has a purple edge crossing a linear number of other purple edges. The figure is not to scale; see the appendix for details.}
\label{fig:MartiniGlass-Preview}
\end{figure}

\section{Points on a line or on a circle}\label{sec:linecircle}

In this section we describe efficient algorithms to find a minimum RBP spanning graph if the points lie on a line or on a circle. In the circle case, we first present additional geometric observations that allow us to use dynamic programming.

\subsection{Points on a line}\label{sec:1-dim}

\begin{figure}[t]
\centering
\includegraphics{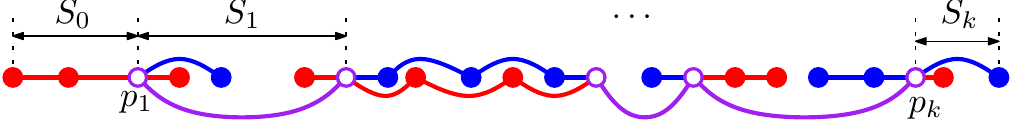}
\caption{A minimum RBP spanning graph of points on a line, and its partition into sets $S_i$ (some edges are depicted by curved arcs for clarity).}
\label{fig:onedimspanning}
\end{figure}

We now consider the case in which the points of $S$ lie on a line. We number the purple points $p_1,\ldots, p_\pp$ from left to right. For every $1\leq i\leq \pp-1$, let $S_i$ be the set of points between $p_i$ and $p_{i+1}$ (this set will have both $p_i$ and $p_{i+1}$, and possibly some blue and/or red points). We also define $S_0$ to contain $p_1$ and all red/blue points to its left, and $S_{\pp}$ to contain $p_{\pp}$ and all red/blue points to its right (see Figure~\ref{fig:onedimspanning}). First, we show that each such subset of points can be treated independently.

\begin{lemma}\label{lem_partition}
Let $S$ be a set of red, blue, and purple points located on a line, and let $G^*$ be a minimum RBP spanning graph of $S$.
Then for every edge $qq'$ of $G^*$, the points $q$ and $q'$ are  in $S_j$, for some $j\in\{0,\ldots k\}$.
\end{lemma}
\begin{proof}
Let $qq'$ be an edge of $G^*$ such that $q \in S_i$, $q' \in S_j$, and $i \neq j$.
Then $qq'$ goes through a purple point $p$.
It is straightforward to verify that we can replace such edge by the pair $qp$ and $pq'$. The new construction has the same weight and also is a RBP spanning graph. Moreover, since it contains one extra edge, it must contain a cycle. Therefore, there is an edge that can be removed while preserving the RBP spanning property. This removal will decrease the total weight of the tree, in contradiction with the optimality of $G^*$.
\end{proof}

Using this lemma it is straightforward to obtain an efficient algorithm.

\begin{theorem}\label{theo-1d}
Let $S$ be a set of $n$ red, blue, and purple points located on a line, given in sorted order. A minimum RBP spanning graph of $S$ can be computed in $O(n)$ time.
\end{theorem}
\begin{proof}
Note that each set $S_i$ has at most two purple points. Thus, there can be at most one purple edge connecting points within each set. We consider two cases: either the edge is present or not. For each case, both the red and blue subproblems can be easily solved in time proportional to the size of $S_i$ (if the edge is not present, $p_i$ and $p_{i+1}$ must be connected  both by red and blue paths. Otherwise, the longest edge in each path could be removed, while preserving connectivity). Out of both possible solutions, we keep the one with minimum weight and proceed to the next value of $i$. By Lemma~\ref{lem_partition}, each subproblem $S_i$ can be solved independently, thus the minimality of each $S_i$ implies that the resulting RBP spanning graph has minimum weight.
\end{proof}

\subsection{Points on a circle}

Recall that, by Proposition~\ref{prop:MartiniGlass}, when the points are in the plane, a purple edge may cross many other purple edges. This fact holds even if the points lie in convex position (see an example with a single crossing in Figure~\ref{fig:spanningexamples} (b); a similar example with more crossings is given in the Appendix). However, below we prove that for points on a circle the situation is structurally different, and purple edges cannot cross any other edges.

\begin{lemma}\label{lem_nocross}
Let $S$ be a set of red, blue, and purple points located on a circle.
A minimum RBP spanning graph of $S$ cannot have a purple edge crossing any other edge.
\end{lemma}
\begin{proof}
Let $G$ be a minimum RBP spanning graph in which two edges $e_1={vv'}$ and $e_2={ww'}$ cross. We will perform a local transformation that will reduce the weight of $G$, contradicting the minimality of $G$.

First assume that both $e_1$ and $e_2$ are purple, and consider the four paths in $G[R\cup P]$ starting at either $v$ or $v'$ and ending at $w$ or $w'$.
Without loss of generality, we can assume that the minimum-link path among the four (i.e., the path with smallest number of edges) connects $v$ and $w$ (illustrated in Figure~\ref{fig_circlecases}).
Let $\pi_R$ be such path; note that, by definition of minimum-link path, $\pi_R$ cannot use edge ${vv'}$, since then the red path connecting $v'$ and $w$ would be one link shorter.
Analogously, $\pi_R$ cannot use edge ${ww'}$ neither.
Let $\pi_B$ be the minimum-link blue path among those that connect $v$ or $v'$ with $w$ or $w'$. We now distinguish three cases depending on the number of shared endpoints between $\pi_R$ and $\pi_B$:

\begin{figure}[t]
\centering
\includegraphics{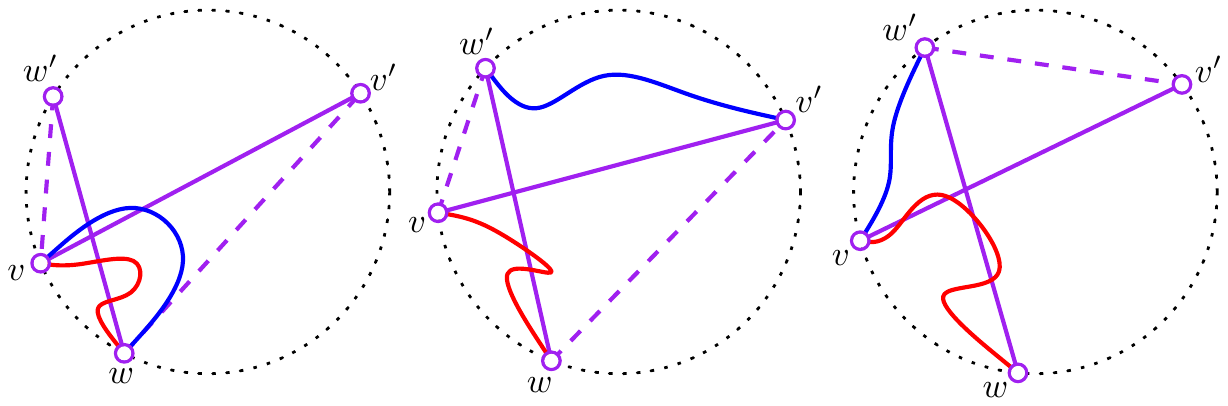}
\caption{The three cases in the proof of Lemma~\ref{lem_nocross}, local transformations are shown by dashed purple edges. No assumptions are made on the number of crossings between $\pi_R$, $\pi_B$, ${vv'}$, and ${ww'}$.
Note that the red and blue paths are drawn schematically---in the actual paths all vertices must lie on the circle.}
\label{fig_circlecases}
\end{figure}

\begin{description}
\item[Paths $\pi_R$ and $\pi_B$ share both endpoints.]
This case is depicted in Figure~\ref{fig_circlecases}, left.
We replace edges ${vv'}$ and ${ww'}$ with edges ${vw'}$ and ${v'w}$. Since we replaced two diagonals of a quadrilateral by two sides, the resulting graph $G'$ has smaller weight than $G$. We now prove that $G'$ is indeed spanning. First consider the red tree in $G'$: the removal of edge ${ww'}$ creates two connected components. Moreover, points $v$ and $w'$ must belong to different connected components (otherwise, the edge ${ww'}$ together with $\pi_R$ and the path from $v$ to $w'$ would create a cycle in $G$). Thus, by adding the edge ${vw'}$ we reconnect the two components again. Likewise, the removal of edge ${vv'}$ creates two red connected components that are reconnected with the edge ${wv'}$. That is, graph $G'$ also spans red. We repeat the same reasoning for blue and obtain that $G'$ is a RBP spanning graph with smaller weight than $G$, a contradiction.
\item[Paths $\pi_R$ and $\pi_B$ share no endpoints.]
This case is depicted in Figure~\ref{fig_circlecases}, center.
We proceed as in the previous case, replacing edges ${vv'}$ and ${ww'}$ by ${vw'}$ and ${v'w}$. The argumentation is identical to the previous case.
\item[Paths $\pi_R$ and $\pi_B$ share one endpoint.]
This case is depicted in Figure~\ref{fig_circlecases}, right.
Without loss of generality, $v$ is the shared endpoint and the other endpoint of $\pi_B$ is $w'$.
In this case, the red path in $G$ from $v'$ to $w$, and the blue path in $G$ from $v'$ to $w'$, both use the edge ${vv'}$.
We can replace this edge by either ${v'w}$ or ${v'w'}$ while  maintaining the spanning property. Using the fact that the four vertices lie on a circle, it is easy to see that either $||{v'w}||<||{v'v}||$ or $||{v'w'}||<||{v'v}||$ must hold, thus one of the two resulting graphs will have smaller weight.
\end{description}

If one of the edges is not purple the situation is easier, since we need to consider only one color. Assume that edge $e_2={ww'}$ is red, and that the red path from $v$ to $w$ uses neither $e_1$ nor $e_2$. Then we can replace edge $e_2$ by either ${vw'}$ or ${v'w'}$ to obtain a graph of smaller weight. Note that, since we are replacing a red edge by another red edge, the spanning property of blue cannot be altered, and the lemma is proved.
\end{proof}

Next we present another crossing property that will be useful for our algorithm.

\begin{corollary}
\label{cor:crosssings}
Let $S$ be a set of red, blue, and purple points located on a circle.
In a minimum RBP spanning graph $G$ of $S$, no red or blue edge of $G$ can 
cross any segment between two purple points.
\end{corollary}
\begin{proof}
Let $p$, $p'$ be any two purple points of $S$ (not necessarily connected by an edge in $G$), and assume that a red edge $rr'$ of $G$ crosses segment ${pp'}$. As in the proof of Lemma~\ref{lem_nocross}, assume that the red path from $p$ to $r$ uses neither $p'$ nor $r'$. Then, edge $rr'$ can be replaced by either $r'p$ or $r'p'$ to obtain a RBP spanning graph of smaller weight.
\end{proof}

Analogously to the 1-dimensional case of points on a line, we number the purple points $p_1,\ldots, p_\pp$ in clockwise order. For any $1\leq i\leq \pp$, let $S_i$ be the set of red and blue points between $p_i$ and $p_{i+1}$. We assume that all indices are taken modulo $k$, thus $p_{\pp+1}=p_1$.
As a consequence of Corollary~\ref{cor:crosssings}, we have the following.

\begin{corollary}\label{cor_indep}
Let $S$ be a set of red, blue, and purple points located on a circle, and let $G$ be a minimum RBP spanning graph of $S$. No edge of $G$ incident to a point in $S_i$ can cross segment $p_ip_{i+1}$, for $1\leq i\leq \pp$.
\end{corollary}

We have proved that each subset $S_i$ can be treated independently when computing a minimum RBP spanning graph of $S$.

\begin{lemma}
\label{lem:circle_deg_at_most_2}
Let $S$ be a set of red, blue, and purple points located on a circle, and let
$G$ be a minimum RBP spanning graph of $S$.
For any purple point $p_i \in S$, let $p_i^*$ be the point on the circle diametrically opposite to $p_i$, and let $C$ be any of the two closed semicircles containing both $p_i$ and $p_i^*$.
It holds that $p_i$ has at most one purple neighbor in $G \cap C$, and thus at most two purple neighbors in total.
\end{lemma}
\begin{proof}
Assume without loss of generality that $C$ is the semicircle to the right of the oriented segment $p_i p_i^*$.
Suppose, for the sake of contradiction, that $p_i$ has at least two purple neighbors in $C$.
Let $p_1$ and $p_2$ be any two of them, as encountered in counterclockwise order from $p_i$ to $p_i^*$ along $C$.
Since $p_i$, $p_1$, $p_2$, and $p_i^*$ are lie in the same semicircle, $p_i p_2$ is longer than  $p_1 p_2$.
Then $G$ can be improved by replacing $p_i p_2$ by $p_1 p_2$, leading to a smaller total length and providing the same connectivity.
It follows that $p_i$ can have at most one purple neighbor in $C$.
\end{proof}

Using these geometric observations we can now present an algorithm to compute a minimum RBP spanning graph of points on a circle.

\begin{theorem}\label{theorem:circle}
Let $S$ be a set of $n$ red, blue, and purple points located on a circle, and angularly sorted. A minimum RBP spanning graph of $S$ can be computed in $O(\pp^3+n)$ time, where $k$ is the number of purple points.
\end{theorem}

\begin{proof}
Analogous to the 1-dimensional case, we number the purple points $p_1,\ldots, p_\pp$ in clockwise order. For any $1\leq i\leq \pp$, let $S_i$ be the set of red and blue points between $p_i$ and $p_{i+1}$. We assume that all indices are taken modulo $k$, thus $p_{\pp+1}=p_1$.

We present an algorithm based on dynamic programming. By Corollary~\ref{cor_indep}, we can find the solution for each set $S_i$ independently. This is a 1-dimensional problem that can be solved in linear time using the approach of Theorem~\ref{theo-1d}. Thus, we solve all subproblems $S_i$ in a preprocessing phase, in overall $O(n)$ time (assuming the points are circularly sorted).

The algorithm uses four tables indexed by pairs of purple points, of size $O(\pp^2)$, defined as follows:
\begin{itemize}
\item $PC[i,j]$ contains the length of a minimum RBP spanning graph  for the subproblem defined by all the points on and to the left of the oriented segment $p_i p_j$ (including $p_i$ and $p_j$), assuming that $p_i$ and $p_j$ are already connected in both the red and the blue subgraphs using points to the right of the oriented segment $p_i p_j$.
In other words, the  minimum RBP spanning graph whose length is stored in $PC[i,j]$ does not need to connect $p_i$ and $p_j$ in neither red nor blue.

\item $RC[i,j]$ contains the length of a minimum RBP spanning graph  for the subproblem defined by all the points on and to the left of the oriented segment $p_i p_j$ (including $p_i$ and $p_j$), now assuming that $p_i$ and $p_j$ are already connected in the red subgraph and disconnected in the blue subgraph. Therefore, the solution associated with $RC[i,j]$ must connect $p_i$ and $p_j$ in its blue subgraph (but not in the red subgraph).

\item $BC[i,j]$ is analogous to $RC[i,j]$, where the roles of red and blue are exchanged.

\item $NC[i,j]$ assumes that $p_i$ and $p_j$ are not connected in any color. Therefore, the solution associated with $NC[i,j]$ must connect $p_i$ and $p_j$ in both red and blue (either using edge $p_i p_j$ or in some other way).
\end{itemize}

With the tables above we can find the length of an optimal solution as
\begin{equation}
\label{eq:min}
\min_{j\leq k} \{ PC[1,j]+NC[j,1] , NC[1,j]+PC[j,1], RC[1,j]+BC[j,1], BC[1,j]+RC[j,1] \}.
\end{equation}
That is, if an optimal solution contains an edge of the form $p_1 p_j$ (for some $j\leq k$), Lemma~\ref{lem_nocross} guarantees that $p_1 p_j$ cannot be crossed by any other edge, giving rise to two independent subproblems, one to the left of the oriented segment $p_1 p_j$, and one to the right, of types $PC$ and $NC$.

Otherwise, $p_1$ cannot have an adjacency with other purple points in an  optimal solution. By Corollary~\ref{cor:crosssings}, no segment $p_1 p_j$ can be crossed. Therefore, each value of $j$ gives rise to two independent subproblems, which now can be of any of the types considered in Formula~\ref{eq:min}.



Thus, it remains to show how to fill in the four tables.
All four tables are filled in a very similar fashion, so in the following we focus on how to compute $\Lambda C[i,j]$ (for some $\Lambda \in\{P,R,B,N\}$).
As is standard in dynamic programming, we fill in the table entries in order, so that when it is time to compute the value of an entry $(i,j)$, all the entries corresponding to \emph{smaller} problems, i.e., corresponding to purple pairs to the left of the oriented segment $p_ip_j$, have already been computed.

The simplest entries are those of the form $(i,i+1)$. To the left of the oriented edge $p_ip_{i+1}$ there are no purple points, thus only red and blue points may be present.
A table entry of this kind can be computed using the approach of Theorem~\ref{theo-1d}, since each solution is independent of the rest. Note that some entries may have no solution (e.g., if there are only red points to the left of the oriented edge $p_ip_{i+1}$, there cannot be a solution for subproblem $RC[i,i+1]$). Whenever no solution exists, the corresponding entry of the table can be filled in with $+\infty$. Note that all entries of the form $PC[i,i+1]$, $RC[i,i+1]$, $BC[i,i+1]$, and $NC[i,i+1]$ can be computed in overall $O(n)$ time, since no red or blue point can appear in more than one set $S_i$.

Without loss of generality, assume $1 \leq i < j \leq \pp$.
In order to fill in entry $\Lambda C[i,j]$ we consider the following cases, and store the one with minimum cost.

\begin{figure}[bt]
\centering
\includegraphics{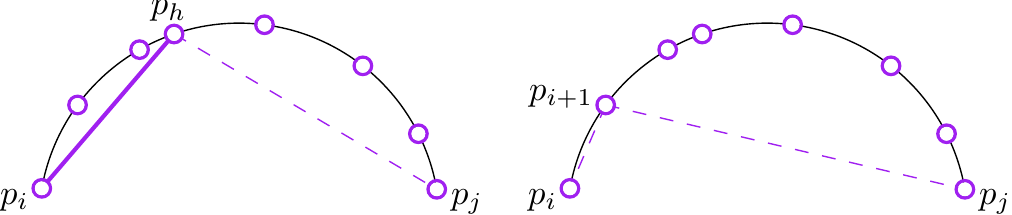}
\caption{Proof of Theorem~\ref{theorem:circle}: cases that appear when filling in entry $\Lambda C[i,j]$ (for some $\Lambda \in\{P,R,B,C\}$.}
\label{fig:DPCircle}
\end{figure}

\begin{itemize}

\item Case I: $p_i$ is connected to some purple point in an optimal solution of the subproblem associated to $(i,j)$.
Let $h$ be the largest such index, see Figure~\ref{fig:DPCircle}
(left).
Recall that vertices are indexed in clockwise order.
In this case, no purple edge can cross segment $p_h p_j$:  by definition of $h$, no edge connects $p_i$ to a vertex with index larger than $h$ in any solution. Hence, by Lemma~\ref{lem_nocross}, no edge can cross $p_i p_h$.
Since $p_i$ and $p_h$ are connected in both red and blue, $p_h$ and $p_j$ must be connected in the same way as $p_i$ and $p_j$. Thus, it follows that in this case $\Lambda C[i,j] = PC[i,h]+\Lambda C[h,j]+||p_ip_h||$.

Since we do not know the value of $h$, we try all possible candidates for $p_h$, and keep the one minimizing the cost.


\item Case II: $p_i$ is not connected to other purple points in any optimal solution of the subproblem associated to $(i,j)$. 
By Corollary~\ref{cor_indep}, in an optimal solution no edge can cross $p_i p_{i+1}$.
In addition, since no purple edge is incident to $p_i$, segment $p_{i+1}p_{j}$ cannot be crossed either, see Figure~\ref{fig:DPCircle} (right).
Therefore, an optimal solution can be found by combining the solutions associated to the pairs $(i,i+1)$ and $(i+1,j)$.
The exact cases to consider depend on the value of $\Lambda$:
\begin{description}
\item[$\Lambda=N$]
In this case it is assumed that $p_i$ and $p_j$ are not connected in any color.
Point $p_{i+1}$ must then connect to both $p_i$ and $p_j$ in each of the subproblems.
Thus, $NC[i,j]=NC[i,i+1]+NC[i+1,j]$.
\item[$\Lambda=R$] In this case $p_i$ and $p_j$ are assumed to be already connected in red. Thus, $p_{i+1}$ can only connect to one of them in red, and must connect to both of them in blue. That is, $RC[i,j]=\min\{NC[i,i+1]+RC[i+1,j], RC[i,i+1]+NC[i+1,j]\}$
\item[$\Lambda=B$] This case is analogous to the previous one: $BC[i,j]=\min\{NC[i,i+1]+BC[i+1,j], BC[i,i+1]+NC[i+1,j]\}$
\item[$\Lambda=P$] In this case $p_i$ and $p_j$ are assumed to be already connected in both colors. Point $p_{i+1}$ must then connect to either $p_i$ or $p_j$ (but not both) in each of the two colors. Thus, $PC[i,j]=\min\{NC[i,i+1]+PC[i+1,j], PC[i,i+1]+NC[i+1,j], RC[i,i+1]+BC[i+1,j], BC[i,i+1]+RC[i+1,j]\}$.
\end{description}

\end{itemize}

Based on the description above, in order to fill in one entry it is necessary to consider up to $O(\pp)$ possible subproblems, each taking constant time to evaluate. Therefore, the total running time of the algorithm is $O(\pp^3+n)$.
\end{proof}

\section{General-case algorithm}~\label{sec:matroid}
In this section we give a polynomial time algorithm for computing a RBP spanning graph, regardless of the position of the points in the plane.
Our solution is based on standard tools from matroid theory~\cite{schrijver-book}.
 The main result that we use is the existence of an efficient algorithm for the {\em weighted matroid intersection} problem, defined as follows: given two matroids with the same ground set $M$ and independent set families $\mathcal{I}_r$ and $\mathcal{I}_b$,
 and a weight function $w: M \rightarrow \mathbb{R}$, the aim is to find the common independent set $\mathcal{I}_r \cap \mathcal{I}_b$ of maximum weight. The algorithm runs in time polynomial on $|M|$, and more importantly, does not depend on the size of $\mathcal{I}_r$ or $\mathcal{I}_b$ (for more information on the weighted matroid intersection problem see~\cite[Section 41.3]{schrijver-book}).

We claim that any minimum RBP spanning graph problem instance can be transformed to a matroid intersection problem instance, and thus it can be solved in polynomial time on the number of points. In order to show this, we consider alternative definition of matroids via base families. For more information on the base families definition of matroids and its equivalence to the independent set formulation see~\cite{schrijver-book}. We start by viewing the red blue purple spanning tree problem from a graph perspective.

Let $R$, $B$, and $P$ denote the sets of red, blue, and purple points, respectively. Let $S = R \cup B \cup P$, and $G=(S,E)$ be the complete undirected graph whose ground set is $R \cup B \cup P$. We also define the weight $w(e)$ of an edge $e \in E$ as its Euclidean length.

We denote by $\mathcal{B}_r \subseteq 2^E$
the family of edge sets of trees in $G$ that span $R \cup P$.
 Similarly, we denote by $\mathcal{B}_b \subseteq 2^E$
the family of edge sets of trees in $G$ that span $B \cup P$.
It is well-known that $(E, \mathcal{B}_r)$ and $(E, \mathcal{B}_b)$ are matroids (known as \emph{graphic matroids}, see~\cite[Section 39.4]{schrijver-book}).

   We also define the following ``basis supersets'':
\begin{eqnarray*}
 {\mathcal{Q}}_r
& = \{X \subseteq E \mid X \mbox{ contains some }B \in \mathcal{B}_r \},\\
  {\mathcal{Q}}_b
& = \{X \subseteq E \mid X \mbox{ contains some }B \in \mathcal{B}_b \}.
\end{eqnarray*}

\noindent
 Then, the minimum RBP spanning graph problem is formulated as follows:
\[
\mbox{(P0) \quad Minimize}  \sum_{e \in X}w(e)
\quad
\mbox{subject to }
 X \in {\mathcal{Q}}_r \cap {\mathcal{Q}}_b,
\]
This formulation is not exactly a matroid intersection problem, but its complement is. Thus, we define $ \overline{\mathcal{Q}}_r = \{Y \subseteq E \mid E \setminus Y \in \mathcal{Q}_r \}$ and $\overline{\mathcal{Q}}_b
 = \{Y \subseteq E \mid E \setminus Y \in \mathcal{Q}_b \}$, and consider the following maximization problem:
\[
\mbox{(P1) \quad Maximize}  \sum_{e \in Y}w(e)
\quad
\mbox{subject to }
 Y \in \overline{\mathcal{Q}}_r \cap \overline{\mathcal{Q}}_b.
\]
Since the complement of a solution for (P0) is a solution for (P1) and {\em vice versa}, both problems are equivalent.
In this setting, $\overline{\mathcal{Q}}_r$ and $\overline{\mathcal{Q}}_b$
are independent set families of some matroids:

\begin{proposition}[{\cite[Section 39.2]{schrijver-book}}]
\label{prop:matroid}
 For any matroid $(E, \mathcal{B})$
with basis family $\mathcal{B} \subseteq 2^E$,
define
\begin{align*}
 {\mathcal{Q}}
& = \{X \subseteq E \mid X \mbox{\rm contains some }B \in \mathcal{B} \},\\
 \overline{\mathcal{Q}}
& = \{Y \subseteq E \mid E \setminus Y \in \mathcal{Q} \}.
\end{align*}
 Then, $(E, \overline{\mathcal{Q}})$ is also a matroid with
independent set family $\overline{\mathcal{Q}}$;
such a matroid is called the \emph{dual matroid} of $(E, \mathcal{B})$.
\end{proposition}
\noindent
In particular, both $\overline{\mathcal{Q}}_r$ and $\overline{\mathcal{Q}}_b$
are independent set families of dual matroids
of $(E, \mathcal{B}_r)$ and $(E, \mathcal{B}_b)$, respectively.  From Proposition~\ref{prop:matroid} it follows that (P1) is an instance of the matroid intersection problem, for which a polynomial-time algorithm exists.

\begin{theorem}\label{theorem:general}
Let $S$ be a set of $n$ red, blue, and purple points
located in the plane. A minimum RBP spanning graph of $S$ can be computed in $O(n^6)$ time.
\end{theorem}
\begin{proof}
Given set $S$, we construct the equivalent (P1) problem instance and solve it using the augmenting algorithm (see \cite[Theorem~41.7]{schrijver-book}). We note that this a is well-known result, thus its correctness is already established. In the following we present an interpretation of the algorithm in terms of the original problem (P0) and bound the running time.

Since any RBP spanning graph is connected, it must have at least $n-1$ edges. Observe that the addition of edges maintains the RBP spanning property. Thus, we conclude that for every $i \in \{n-1, \ldots, |E|\}$ there exists an RBP spanning graph with exactly $i$ edges. Recall that, in our setting, two points are connected in $E$ if and only if they have the same color or one of the two is purple. In particular, we have $|E|= {|R|+|P| \choose 2}+ {|B|+|P| \choose 2} - {|P| \choose 2} \leq {|S| \choose 2} \in O(n^2)$.

Starting from $i=|E|$ and continuing with smaller values of $i$, the algorithm finds an RBP connected subgraph $X_i$  with $i$ edges that minimizes the weight $w(X_i)$ (among those that contain exactly $i$ edges). Then, the subgraph $X_i$ of overall minimum weight will be a minimum RBP spanning subgraph.

\medskip

\noindent
{\bf Step 0:}
Set $i=|E|$ and $X_i := E$.
\\
{\bf Step 1:}
 Take a sequence of edges $e_1, e_2, \ldots, e_{2h+1}$ $(h \geq 0)$
such that
\begin{align*}
& e_1, e_3, \ldots, e_{2h+1} \in X_i,
\\
& e_2, e_4, \ldots, e_{2h} \in E \setminus X_i,
\\
& X_i - e_{1} \mbox{ induces a connected subgraph for the blue and purple points},
\\
& X_i - e_{2j-1} + e_{2j} \mbox{ induces a connected subgraph for the red and purple points }
j=1,2,\ldots, h,
\\
& X_i + e_{2j} - e_{2j+1} \mbox{ induces a connected subgraph for the blue and purple points for }
j=1,2,\ldots, h,
\\
& X_i - e_{2h+1} \mbox{ induces a connected subgraph for the red and purple points}.
\end{align*}

Such a sequence will exist provided that $i\geq n$ (i.e., we can remove an edge in a cycle or add a purple edge and remove both a red and a blue edge). Among such sequences, take the one with
the minimum value of
$ -\sum_{j=0}^{h}w(e_{2j+1}) + \sum_{j=1}^{h}w(e_{2j}). $
If there exist many such sequences, then take the one with
the minimum number of edges.
\\
{\bf Step 2:}
 Set
$ X_{i-1}:= X_i \setminus \{e_1, e_3, \ldots, e_{2h+1}\} \cup
\{e_2, e_4, \ldots, e_{2h}\}$
and $i:=i-1$.
 Go to Step 1.

\medskip

 Due to the choice of a sequence in Step 1, we can show that
if $X_i$ is an RBP connected subgraph  with $|X_i|=i$
with the minimum weight among those with exactly $i$ edges, then
$X_{i-1}$ will be an RBP connected subgraph with $i-1$ edges
having the minimum weight among those with $i-1$ edges (see \cite[Section 41.3]{schrijver-book}).
 Step 0 initializes with the unique graph with exactly $|E|$ edges. Since it is unique, it has the minimum weight, and thus each $X_i$ is an RBP connected subgraph with minimum weight among those containing $i$ edges.

We analyze the time complexity of the algorithm above. Let $n=|V|$ and $m=|E|$.
 The number of iterations of the algorithm is $m-(n-1)\in O(n^2)$.
 The sequence of edges in Step 2 can be computed
by solving a single-source shortest path problem
on an auxiliary edge-weighted graph. Thus, overall the algorithm terminates in
$O(m \cdot \mbox{SSP}(N,M))$ time, , where $N$ (resp., $M$) denote the number of vertices (resp., edges) in the auxiliary edge-weighted graph, and $\mbox{SSP}(n',m')$ denotes the time complexity for
solving a single-source shortest path problem on a graph
with $n'$ vertices and $m'$ edges. 

In our case, $m\in O(n^2)$, $N=O(m)=O(n^2)$, $M=O(m^2)=O(n^4)$ and we can solve the shortest path problem using Dijsktra's algorithm in $O(M+N\log N)=O(n^4)$ time.\footnote{Note that in the auxiliary graph some edges may have negative weight. Thus, the straightforwards Dijkstra's approach cannot be used. Nevertheless, this difficulty can be avoided by assigning a potential to each node. For more details see~\cite{schrijver-book}.} Thus, the overall running time is bounded by $O(n^6)$, as claimed.
\end{proof}

\section{Approximation}\label{sec:approx}

Given the high running time of the exact algorithm presented in the previous section, it is worth considering more efficient approximation algorithms.
A simple approximation algorithm determines the red-purple
minimum spanning tree and the blue-purple minimum spanning tree independently,
and takes the union of their edges. It is easy to see that this
is a $2$-approximation algorithm that requires $O(n\log n)$ time.

Another approximation algorithm, $A$, starts by computing the
minimum spanning tree of the purple edges, and then adds the
red and blue points in an optimal manner in the style of Kruskal's
algorithm for minimum spanning trees. Algorithm $A$ can also
be implemented to run in $O(n\log n)$ time by computing the
Delaunay triangulation of the red and purple points and of the blue and purple points.
%
%
It is easy to argue that $A$ also is a $2$-approximation algorithm
but interestingly, we can prove a better bound (close to $1.6$)
by expressing the approximation factor in the Steiner ratio $\rho$ (the ratio between the length of a minimum spanning tree and the length of a  minimum Steiner tree).
Gilbert and Pollak~\cite{gp-smt-68} conjectured that
$\rho=\frac {2}{\sqrt{3}}\approx 1.15$, but this conjecture has not been proved
yet.\footnote{A proof of the conjecture by Du and Hwang,
``A proof of Gilbert-Pollak Conjecture on the Steiner ratio'', \emph{Algorithmica} 7:121--135 (1992), turned out to be incorrect.}
Chung and Graham~\cite{cg-nbesmt-85} showed a bound of $\approx 1.21$,
which is currently the best-known upper bound on $\rho$.

\begin{theorem}
Approximation algorithm $A$ is a $(\frac 12\rho+1)$-approximation of the
minimum RBP spanning graph, where $\rho$ is the Steiner ratio.
\end{theorem}
\begin{proof}
Let $R$, $B$, and $P$ be sets of red, blue, and purple points.
Let $G^*$ be a minimum RBP spanning graph.
Let $R^*$ be the red edges, $B^*$ the blue edges, and $P^*$ the purple edges in $G^*$.
Recall that  $A$ computes a spanning graph by taking the
minimum spanning tree of the purple points, and then adding the red and
blue points optimally. We denote the resulting graph on $R\cup B\cup P$
by $G'$, and its red, blue, and purple edges by $R'$, $B'$, and $P'$.



Suppose first that $G^*$ has no purple edges.
Then algorithm $A$ gives extra length in terms of purple
edges equal to the MST of the purple points, denoted $||P'||$.
The optimal graph $G^*$
must connect all purple points through a red spanning tree
and through a blue spanning tree whose lengths are $||R^*||$ and $||B^*||$.
Algorithm $A$ has a total length of red edges of $||R'||\leq ||R^*||$
and a total length of blue edges of $||B'||\leq ||B^*||$.
Hence the approximation ratio of $A$ in case of absence of purple edges in
the optimal solution is
\[\frac{||P'||+||R'||+||B'||}{||R^*||+||B^*||}\leq
\frac{||P'||+||R^*||+||B^*||}{||R^*||+||B^*||}\,.\]
Since $R^*$ connects all purple points, $||R^*||\geq \SMT(P)$, and similarly, $||B^*||\geq \SMT(P)$, where
$\SMT(P)$ is the length of the Steiner Minimum Tree of $P$.
The ratio is maximized when $||R^*||$ and $||B^*||$ are as small
as possible, so the ratio is upper-bounded by
\[\frac{\MST(P)+2\cdot \SMT(P)}{2\cdot \SMT(P)}
\leq \frac{\rho\cdot\SMT(P)+2\cdot \SMT(P)}{2\cdot \SMT(P)}
= \frac{\rho+2}{2}=\frac 12 \rho+1 \,,\]
where $\MST(P)$ is the length of the Minimum Spanning Tree of $P$.

Next, suppose that $G^*$ has a set $P^*$ of purple edges, and assume them fixed.
We will reason about sets of red, blue, and purple points for which the
algorithm $A$ performs as poorly as possible in terms of approximation ratio.

If $G^*$ has any red point $r$ that has a single red edge incident to it in $R^*$,
then this edge will connect $r$ to the closest red or purple point, otherwise
$G^*$ is not optimal. Algorithm $A$ will choose exactly the same edge in its
solution. Hence, the approximation ratio of $A$ for the points
$R\setminus\{r\}$, $B$, and $P$ is higher than for the points $R$, $B$, and $P$.
The same is true for a blue or purple point that has a single incident edge
in $G^*$. So we can restrict ourselves to analyzing point sets whose
optimal solution does not have any leafs in $G^*$.

The edges of $P^*$ partition the purple points of $P$ into a
number of purple components which are connected by a red spanning forest and
a blue spanning forest, and by the observations above, there are no other red or blue points
in $G^*$. We have
\[||P'||\leq ||P^*|| + \rho ||R^*||\;\mbox{ and }\; ||P'||\leq ||P^*|| + \rho ||B^*||\;,\]
because in $G^*$ the red (blue) connections between the purple components
cannot be shorter than the Steiner Minimum Forest of the purple components.
We can assume by symmetry that $||R^*||\leq ||B^*||$.
The approximation ratio then is
\[\frac{||P'||+||R'||+||B'||}{||P^*||+||R^*||+||B^*||}\leq
\frac{||P'||+||R^*||+||B^*||}{||P^*||+||R^*||+||B^*||}\leq
\frac{||P'||+ 2||R^*||}{||P^*||+2||R^*||}\leq
\frac{||P^*||+(2+\rho) ||R^*||}{||P^*||+2||R^*||}\,.\]
This ratio is maximal when $||P^*||=0$, in which case we get
exactly the same ratio as above, when no purple edges are present.
\end{proof}

\begin{figure}[tb]
\centering
\includegraphics{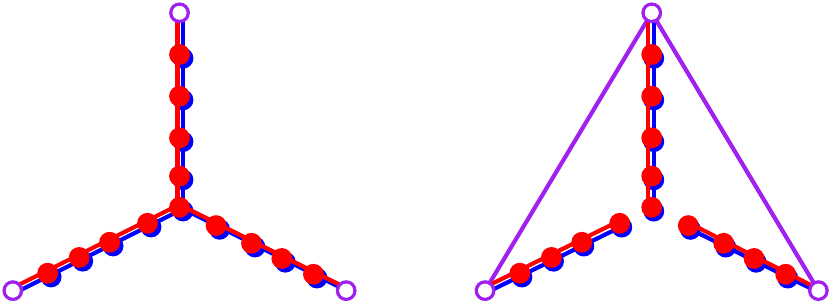}
\caption{A minimum RBP spanning graph and the RBP spanning graph
on the same points obtained by approximation algorithm $A$, resulting in an approximation factor that tends to $1+\frac{1}{\sqrt{3}}$.}
\label{fig:steiner}
\end{figure}
We note that $A$ cannot be a $c$-approximation for any $c< 1+\frac{1}{\sqrt{3}} \approx 1.58$ (a counterexample is shown in Figure~\ref{fig:steiner}). Hence our approximation analysis is tight if the Gilbert-Pollak conjecture is true. It is possible that a PTAS exists for our problem, but it is not clear
whether the techniques of Arora~\cite{a-ptase-98} or Mitchell~\cite{m-gsaps-99}
for the Euclidean traveling salesperson problem can be applied, since RBP spanning graphs are not planar, and the number of crossings of a single edge can be large.

\section{Extensions and Future Work}\label{sec:extensions}
{\bf Line drawings.} Motivated by LineSets~\cite{alper11}, we could extend the ideas of this paper as follows: given the set of red, blue, and purple points, we now wish to compute a minimum RBP spanning graph such that the subgraphs induced by the red and blue sets contain a spanning paths. That is, the red and purple edges form a path connecting all red and purple points, and the blue and purple edges form a path connecting all blue and purple points (see Figure~\ref{fig:extensions} (a)).

If we are interested in minimizing the total length, the problem becomes NP-hard, since it can be seen as a generalization of the traveling salesman problem (TSP). Nonetheless, we can obtain a $(2+\eps)$-approximation by independently computing an approximate TSP for the blue and purple points and for the red and purple points, and simply taking the union. An approach similar to the spanning tree case seems to fail and hence a better solution remains an open problem. The question also remains open for points on a line or a circle.

\smallskip\noindent
{\bf Beyond purple.} In this paper we considered the case where there are exactly two types of points, leading to an input with red, blue, and purple points. In general, we might have $k$ different sets, all denoted by primary colors. For instance, for $k=3$ we could have red, blue, and yellow sets, which leads to three secondary colors (purple, orange, green) and one tertiary color (black). The objective is again to minimize the total length of a multi-colored spanning graph that has the property that the subgraphs induced by the red, blue, and yellow sets are connected (see Figure~\ref{fig:extensions} (b)).

The $2$-approximation immediately generalizes to a $3$-approximation (or a $k$-approximation for $k$ primary colors).
We can improve on this by incorporating our $(1+\frac12\rho)$-approximation algorithm to obtain a $(2+\frac12\rho)$-approximation for three sets, or more generally a $(\lceil \frac12k \rceil + \lfloor \frac12k \rfloor \frac12\rho)$-approximation for $k$ sets.
Interestingly, our algorithms for points on a line or on a circle are not straightforward to generalize; these problems remain open.

\begin{figure}[h]
\centering
\subfigure[]{{\includegraphics[scale=.7]{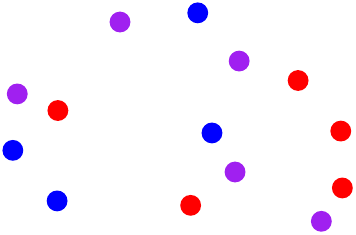}}
\hspace{10pt}
{\includegraphics[scale=.7]{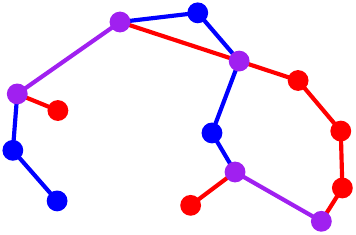}}}
\hfill
\subfigure[]{{\includegraphics[scale=.7]{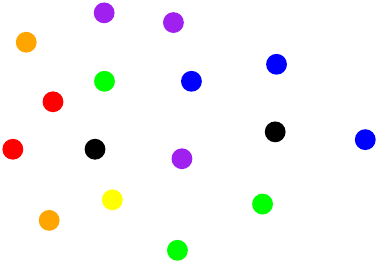}}
\hspace{10pt}
{\includegraphics[scale=.7]{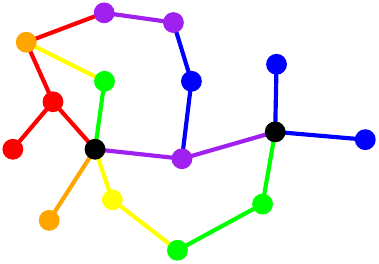}}}
\caption{(a) A set of red, blue, and purple points, and a graph that connects all red points in a path and all blue points in a path.
(b) A set of multicolored points representing red, blue, and yellow sets, and a corresponding spanning graph.
  }
  \label{fig:extensions}
  \vspace{-.5\baselineskip}
\end{figure}

\section*{Acknowledgements}
F. H., M. K., V. S., and R.~I.~S. were partially supported by ESF EUROCORES programme EuroGIGA, CRP ComPoSe: grant
EUI-EURC-2011-4306, and by project MINECO MTM2012-30951. F. H., V. S., and R.I. S. were  supported by projects Gen. Cat. DGR 2009SGR1040 and DGR 2014SGR46.
M.~K. was supported by ELC project (MEXT KAKENHI No. 24106008). M.L. and B.S. were supported by the Netherlands Organisation for Scientific Research (NWO) under project numbers 639.021.123 and 639.023.208, respectively.
A.~S. is supported by JSPS/MEXT KAKENHI Grant Numbers 24500002, 25106503. R.~I.~S. was funded by Portuguese funds through CIDMA and FCT, within project PEst-OE/MAT/UI4106/2014, by FCT grant SFRH/BPD/ 88455/2012, and by Spanish MINECO through the Ram{\'o}n y Cajal program. T.~T. is supported by ELC project of Grant-in-Aid for Scientific Research on Innovative Areas 24106007, Grant-in-Aid for Scientific Research (B) 40312631, Grant-in-Aid for Exploratory Research 24650001, JSPS Grant Scientific Research (B) 15H02665, MEXT Japan, and Kawarabayashi Big Graph ERATO project, Japan Science and Technology Agency. The authors would like to thank Hugo Alves Akitaya for his input in the analysis of the algorithm in Section 4.
\small

\bibliographystyle{abbrv}
\bibliography{redbluepurplerefs}

\clearpage
\normalsize
\appendix
\section{Proof of Proposition~\ref{prop:MartiniGlass}}
In this section, we present a problem instance for which the optimal solution contains a purple edge that is crossed by $\Theta(n)$ other purple edges. Before giving the exact construction, we present a relation between MST and RBP spanning trees, which we find interesting on its own. Then we describe the construction, we argue about parts of it, and finally prove that this construction indeed satisfies our claim.

For simplicity in the exposition, in this section we assume that the points are in {\em general position}: we require that there exist no two distinct pairs of points whose distances are equal (i.e., $\forall u,v,u',v'\in R\cup B\cup P$ it holds that $d(u,v)=d(u',v') \Leftrightarrow \{u,v\}=\{u',v'\}$).
This general position assumption implies uniqueness of minimum spanning trees, and thus greatly simplifies the statements below.
To satisfy the general position requirement, it suffices to do a symbolic perturbation of the input.
Alternatively, the general position requirement can be dropped if the statement of Lemma~\ref{lem_edgesmst} is slightly weakened to existence (rather than universality).

Given a problem instance $(R,B,P)$, let $T_R=MST(R\cup P)$, and $T_B=MST(B\cup P)$. Clearly, not all edges of $T_R$ or $T_B$ need to be present in $G^*$ (see examples in Figure~\ref{fig:spanningexamples}). However, in the next lemma we show that most of them will.

We say that a subset $S'$ of $R \cup P$ is {\em red-maximal} if the subgraph of $T_R$ induced by $S'$, denoted $T_R[S']$, satisfies the following properties:
\begin{enumerate}
\item $T_R[S']$ is connected,
\item all purple points of $S'$ are leaves of $T_R[S']$,
\item any leaf of $T_R[S']$ that is not purple is also a leaf of $T_R$.
\end{enumerate}

Intuitively speaking, a red-maximal subset $S'$ is a maximal subset in which all points are connected in $T_R[S']$ through paths that never traverse purple points. The definition of \emph{blue-maximal} is analogous.


\begin{lemma}\label{lem_edgesmst}
If $S'$ is a red-maximal (resp. blue-maximal) subset with $k$ purple points of a problem instance, then the number of edges of $T_R[S']$ (resp. $T_B[S']$) that are not in $G^*$ is less than $k$.
\end{lemma}

\begin{figure}[b]
\centering
\includegraphics{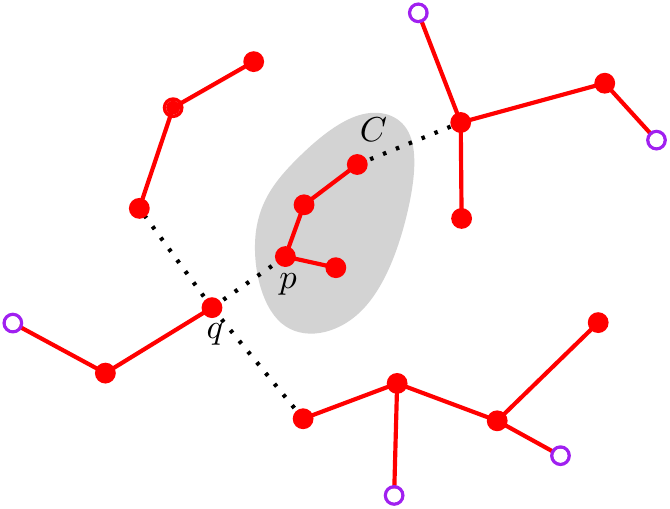}
\caption{Proof of Lemma~\ref{lem_edgesmst}: if many edges of $T_R$ are missing in $G^*$, we can replace one edge and obtain a RBP spanning graph of smaller weight. Edges of $G^*\cap T_R[S']$ are marked in solid, edges of $E$ are dashed, $C$ is the gray region.}
\label{fig-edgesmst}
\end{figure}

\begin{proof}
We prove the lemma for red-maximal subsets.
Let $E$ be the set of edges that are present in $T_R[S']$ but not in $G^*$. In the following we give a process to reduce the size of $E$ whenever it has $k$ or more edges. Since $T_R[S']$ is a tree, the removal of all edges in $E$ will create at least $k+1$ components. In particular, there must be a component $C$ that has no purple edge. Let $p\in C$, $q \in S'\setminus C$ be the two vertices that minimize the distance between points of $C$ and those in $S'\setminus C$. It is well known~\cite{Pre85} that $pq\in T_R[S']$. Moreover, by definition of $q$, the edge $pq$ cannot be in $G^*$, thus $pq\in E$ (see Figure~\ref{fig-edgesmst}).

Our aim is to replace some edge of $G^*$ by the edge $pq$. Recall that $G^*[R \cup P]$ forms a tree, hence there exists a unique path $\pi=(p=v_0, \ldots, v_m=q)$ connecting $p$ and $q$ in the red subgraph of $G^*$. Let $i$ be the smallest index such that $v_i\not\in C$ (note that $i<m$ since $pq\not \in G^*$). We claim that replacing edge $v_{i-1}v_i$ by edge $pq$ in $G^*$, we obtain another RBP spanning graph of smaller weight, giving a contradiction.

By definition of $p$ and $q$, and using the general position assumption, we have $||pq|| < ||v_{i-1}v_i||$. Thus, it remains to prove that this modification indeed preserves the RBP spanning property. Recall that, by hypothesis, all points of $C$ are red, thus the removal of edge $v_{i-1}v_i$ can only affect the spanning property of the red subgraph of $G^*$. By removing this edge, $p$ and $q$ must belong to two different components (since that edge was part of the only path connecting them). Thus, by adding the edge $pq$ we reconnect the two components as desired.
\end{proof}

\paragraph{Description of the construction}
The construction has two parts.
At a global level it consists of points located as illustrated in Figure~\ref{fig:MartiniGlass-Schematic} (left).

\begin{figure}[tp]
\centering
\includegraphics[page=2]{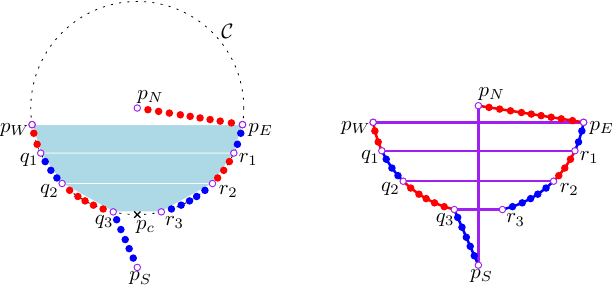}
\caption{Schematic view of the construction (left) and its optimal solution (right), with three levels. Note that this is not to scale. The actual proportions are shown in Figure~\ref{fig:MartiniGlass}.}
\label{fig:MartiniGlass-Schematic}
\end{figure}

Let $p_N=(x_0,y_0)$ be any point in the plane, let ${\cal C}_0$ be the unit circle centered at $p_N$,  and let $p_c$ be the bottom-most point of ${\cal C}_0$ (i.e., $p_c=(x_0,y_0-1)$).
Now define another unit circle $\cal C$, with center $(x_0,y_0-\eps_0)$, for a sufficiently small $\eps_0>0$.
We give more details on how to choose $\eps_0$ later in this section.
Note that the lower semicircle of $\cal C$ lies outside ${\cal C}_0$.
We place most points on the lower semicircle of $\cal C$, as shown in the figure.
Let $p_W$ and $p_E$ be two points with the same $y$-coordinate on the lower semicircle of $\cal C$.
Their exact position is not really important, although some remarks about it are given later.


In the construction we place $2(m+1)$ purple points on $\cal C$, in pairs with the same $y$-coordinate, for any odd value of $m \geq 1$.
On the arc of $\cal C$ going counterclockwise from $p_W$ towards $p_c$, purple points are named $p_W=q_0, q_1, \dots, q_m$, while on the opposite arc, they are named $p_E=r_0, r_1, \dots, r_m$.

Four ``consecutive'' purple points $q_i,r_i, q_{i+1}, r_{i+1}$ form a trapezoid defining what we call a \emph{level}. There are $m$ levels in total.
For each level, on the circular arc of $\cal C$ going counterclockwise from $q_i$ to $q_{i+1}$ (and, respectively, from $r_{i+1}$ to $r_{i}$) we place a chain of red or blue points, one color per side (see Figure~\ref{fig:MartiniGlass-Schematic}).
Each chain has its points very close to each other, so that in any optimal solution they will be connected by edges of their same color. For the sake of simplicity, it is possible to think of the two chains as having the same number of points, although this is not necessary for the construction to work.
Colors between consecutive levels alternate sides.

The exact location of $q_i$ and $r_i$ needs to be chosen carefully to guarantee that the optimal RBP spanning tree is as needed. The main properties of the construction are the following:
(i) $p_N$ is closer to $q_i$ and $r_i$ than to $q_{i+1}$ and $r_{i+1}$, and
(ii) edge $q_i q_{i+1}$ is longer than half of the edge $q_i r_i$.
To guarantee the latter, $q_{i+1}$ is placed at the intersection of two circles: $\cal C$ and the circle centered at $q_i$ with radius half of the length of $q_i r_i$ plus $\eps$, for a small constant $\eps>0$.
This latter property is essential to argue about the optimal spanning graph of a single level.
Figure~\ref{fig:MartiniGlass} shows the construction with proportions that satisfy these requirements.
More precisely, it shows the first level. Additional levels are nested recursively in the bottom side of the trapezoid, using line segment $q_1 r_1$ instead of $p_W p_E$.

\begin{figure}[h]
\centering
\includegraphics{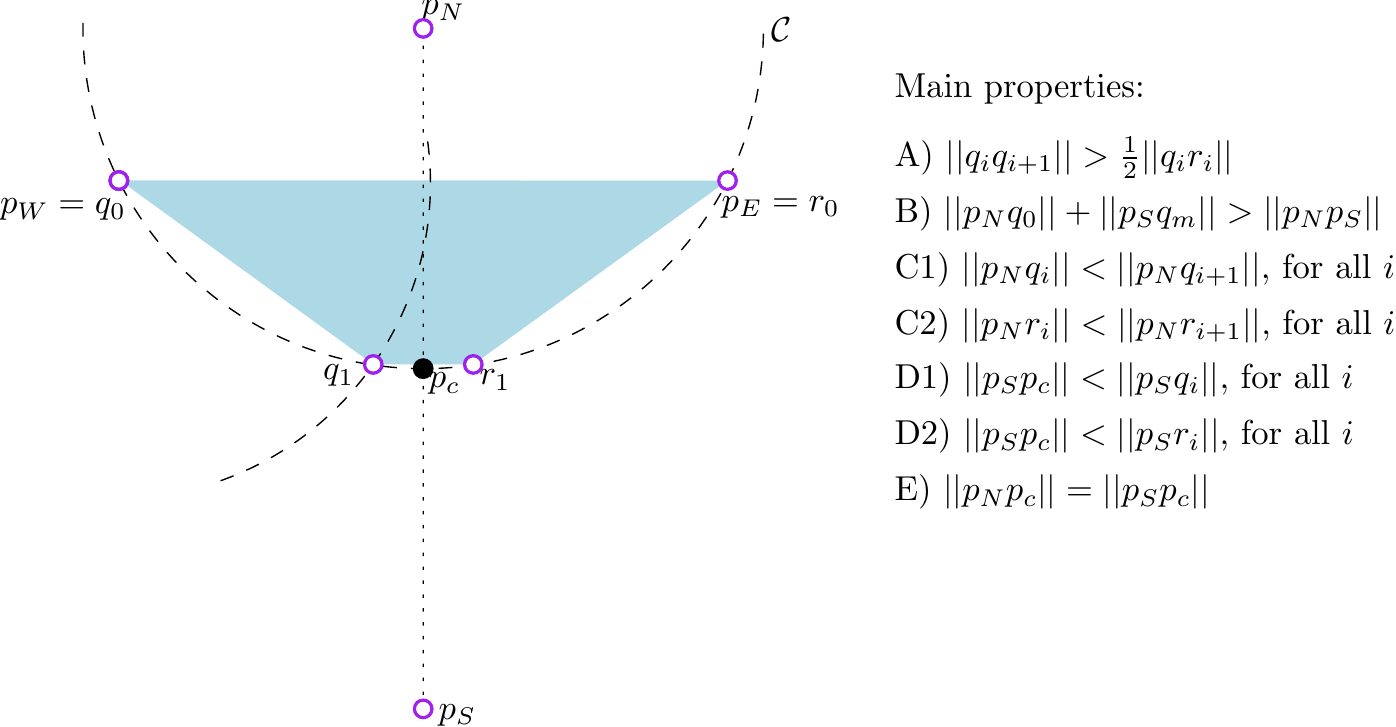}
\caption{One single level of the main construction drawn with the right proportions.
The key properties of the construction are shown on the right.
More levels can be added recursively below line segment $q_1 r_1$.
}
\label{fig:MartiniGlass}
\end{figure}

In this way, if $\eps_0$ is small enough it holds that $p_c$ lies below all the four points of the trapezoids of all levels.

Once the $m$ levels have been added, we insert an additional purple point $p_S$ located vertically below $p_c$, so that $||p_cp_S||=||p_Np_c||$, i.e., at position $(x_0, y_0 - 2-2\eps_0)$. By construction of the points $q_i$ and $r_i$, we have $||p_N\nu_i||+||p_S\eta_j||>||p_Np_S||$ for any $\nu,\eta\in\{q,r\}$
 and $0 \leq i,j\leq m$. 
We complete the construction with  two chains of red and blue points connecting $p_N$ to $p_E$, and $q_m$ to $p_S$, respectively.

One last technical detail concerns the choice of $\eps_0$.
The only condition on this value is that it should be small enough as to guarantee that $y(q_m) > y(p_c)$.
Given that by construction $\eps_0=0$ always leaves a gap between $y(q_m)$ and $y(p_c)$, it is always possible to choose a suitable $\eps_0>0$ maintaining that condition.

Figure~\ref{fig:MartiniGlass-Schematic} (right) shows a RBP spanning tree that, we claim, is the only optimal solution for this configuration of points.
Note that the vertical purple edge is crossed by all the other horizontal purple edges, thus proving Proposition~\ref{prop:MartiniGlass}.

In order to prove Proposition~\ref{prop:MartiniGlass}, we first argue about the optimal RBP spanning graph for the subproblem of the points within the levels.

\begin{lemma}\label{lem:portion-of-martini}
The optimal solution for all points with $y$-coordinates between those of $p_W$ and $q_m$ consists of all horizontal purple edges $q_i r_i$ plus all chains of short red/blue edges, as shown in Figure~\ref{fig:MartiniGlass-Schematic} (right).
\end{lemma}

\begin{proof}
First we argue about one single level, defined by points $q_i$, $r_i$, $q_{i+1}$, $r_{i+1}$.
Refer to Figure~\ref{fig:MartiniGlass-Schematic} (right).
One of the properties of the construction is that $||q_i q_{i+1} || > \frac{1}{2} ||q_i r_i ||$, implying $||q_i q_{i+1} ||+||r_i r_{i+1} || > ||q_i r_i||$.
It is easy to verify that the shortest RBP spanning tree for one level is as claimed: it uses both horizontal purple edges and connects all red/blue points on the circular arcs.

To argue about all the levels, we apply an inductive argument bottom-up.
Suppose the claimed solution is optimal from level $m$ up to level $i$.
Now consider the problem from level $m$ up to one more level, i.e., $i-1$.
Assume that the points between $q_i$ and $q_{i-1}$ are blue, and the ones between  $r_i$ and $r_{i-1}$ are red.

Let $S'$ be the set containing $q_{i-1}$, $q_{i}$, and all blue points in between. By construction, the MST of $S'$ is a path whose endpoints are $q_{i-1}$ and $q_{i}$, and in particular $S'$ is blue-maximal. By Lemma~\ref{lem_edgesmst}, all of the edges of the path (except for at most one) will be present in any RBP spanning graph. Thus, we can connect $q_i$ and $q_{i-1}$ in blue by adding at most one edge whose length is arbitrarily small. Likewise, we can assume that $r_i$ and $r_{i-1}$ are connected through their intermediate red points. If the resulting solution does not need one of those edges, then we can remove it, but the total cost will only change by an arbitrarily small amount. Consider an optimal solution $G^*$ up to level $i-1$. We analyze two cases, depending on whether segment $q_i r_i$ is crossed by some edge in $G^*$.

If $q_i r_i$ is not crossed in $G^*$, then we can partition the problem into two independent subproblems: the trapezoid defined by the four purple points $q_i,r_i, q_{i-1}, r_{i-1}$, and the points located on or below the segment $q_ir_i$. Since in both sub-problems the solution has the desired shape, we derive that the union will also have the desired shape.

Thus, it remains to consider the case in which segment $q_i r_i$ is indeed crossed in $G^*$.
By Corollary~\ref{cor:crosssings}, it can only be crossed by a purple edge, which must connect $r_{i-1}$ or $q_{i-1}$ to some purple point below $q_i r_i$.
Assume without loss of generality that the crossing edge is of the form $\nu_j r_{i-1}$, where $\nu$ is either $q$ or $r$, and $j>i$ (see Figure~\ref{fig_lowerBound}). Recall that all points in the construction (other than $p_N$ and $p_S$) lie on a semicircle. Thus, by Lemma~\ref{lem:circle_deg_at_most_2}, $r_{i-1}$ cannot be adjacent to another purple point in $G^*$. Moreover, by Corollary~\ref{cor:crosssings}, $r_{i-1}$ cannot be adjacent to any red points other than those between $r_{i-1}$ and $r_{i}$, and cannot be adjacent to any blue point.

\begin{figure}[tb]
\centering
\includegraphics{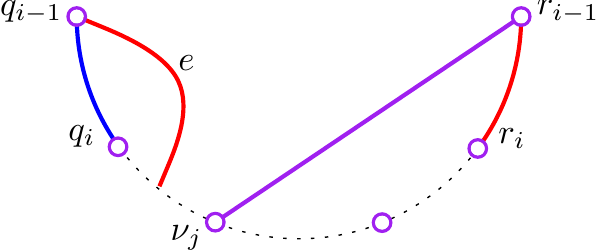}
\caption{If $r_{i-1}$ does not connect to $q_{i-1}$ in $G^*$, then $q_{i-1}$ must also be incident to a long edge $e$ that is useful only in the red tree. Thus, we can replace both edges by $q_{i-1}r_{i-1}$ and obtain a tree of smaller weight.}
\label{fig_lowerBound}
\end{figure}

In particular,
the path connecting $q_{i-1}$ with $r_{i-1}$ must use edge $\nu_jr_{i-1}$ on both red and blue trees of $G^*$. Since $G^*$ is connected in the red subgraph, there must be an edge $e$ in $G^*$ that connects $q_{i-1}$ to some other red vertex. In principle, the other vertex of $e$ could also be purple, but notice that such edge would only be needed for the red path (since we assumed that $q_{i-1}$ can connect to $q_i$ through the blue points in between). The nearest possible vertex is $q_{i}$, thus we conclude that the length of $e$ is at least $||q_{i-1} q_i||$. Similarly, we have $||\nu_j r_{i-1}|| > ||r_{i-1} r_i||$ by the way in which points $q_i$ and $r_i$ were placed.

Recall that one of the key properties of our construction is that  $||q_{i-1} q_i||+||r_{i-1} r_i|| > || q_{i-1} r_{i-1} ||$. Hence, if we replace those two edges by $q_{i-1} r_{i-1}$ we obtain a graph of smaller weight. Moreover, this change cannot affect the RBP spanning property: after we removed the edge $\nu_jr_{i-1}$, vertices $q_{i-1}$ and $r_{i-1}$ belong to different components in both red and blue subtrees. The new purple edge reintroduces the connectivity, thus we would obtain a RBP spanning graph of smaller weight (a contradiction with the fact that the edge $q_ir_i$ is not used in a RBP spanning graph). 
\end{proof}

%

To complete the proof of Proposition~\ref{prop:MartiniGlass}, we consider now the full construction, including $p_N$, $p_S$, and the red and blue points between them and $p_E$ and $q_m$, respectively.

%
\begin{proof} (of Proposition~\ref{prop:MartiniGlass})
Consider the complete construction, as shown in Figure~\ref{fig:MartiniGlass-Schematic} (left).
Recall that $\cal C$ is a unit disk.
We show below that the extra points do not affect the optimal solution for the central construction. Note that in the construction, we can choose the $y$-coordinates of $p_W$ and $p_E$ as low as desired provided that they are above $y(p_c)$. Moreover, notice that the cost of the main construction associated with  our solution (i.e. disregarding points $p_N$, $p_S$, and the two paths connecting these points) decreases monotonically as we lower the $y$-coordinate of $p_W$ and $p_E$. Thus, by placing $p_W$ and $p_E$ with a sufficiently low $y$-coordinate, the cost of our construction can become arbitrarily small.

Now consider the extra points: observe that the cheapest way to connect all red points in the path from $p_N$ to $p_E$ is to form a (red) path between them. Likewise, we obtain the same result for the blue points from $p_S$ to $q_m$. The cost of each of these paths is roughly 1 (slightly larger due to $\eps_0$).
Since we placed all points on the lower semicircle of $C$, outside ${\cal C}_0$, the purple points with highest $y$-coordinates, $p_W$ and $p_E$, are the ones closest to $p_N$.
Therefore $p_N$ must be incident to an edge of length at least $|| p_N p_E ||$ (or $p_N$ will not be connected in the blue subgraph). Further, notice that $p_N$ must be a leaf in the blue subgraph (and $p_S$ in the red subgraph). This is due to the fact that all other points are at distance less than 1 from each other, thus if $p_N$ had degree two or more, one of the adjacencies could be replaced to obtain a graph of smaller length.

Recall that we placed our purple points so that $||p_N\nu || + ||p_S\eta ||  >  ||p_N p_S ||$ for any two purple points $\eta$, $\nu$, thus the edge $p_Np_S$ is shorter than any other combination of edges emanating from $p_N$ and $p_S$.
Any solution that adds extra edges to $p_N$ cannot be of minimal weight, since has extra purple edges.
Now we apply Lemma~\ref{lem_edgesmst} to the set containing $p_E$, $p_N$, and all red points in between (similarly, to $q_m$, $p_S$,  and all blue points in between). These sets are red- and blue-maximal, respectively, hence by Lemma~\ref{lem_edgesmst} any optimal solution must connect them in red and blue paths, respectively.
As before, at most one edge may be missing in each of the paths; for simplicity in the explanation, we keep such edge, since this can only create an arbitrarily small increase in the total cost.
That is, $p_N$ and $p_S$ must have exactly one additional edge (if the red and blue paths are connected) or two edges (otherwise), and the cheapest way of doing so is by connecting $p_N$ to $p_S$ directly.
This construction has minimal weight, and is RBP spanning, thus the statement is shown.
\end{proof}

\end{document}